\documentclass[11pt]{amsart}
\usepackage{mypackage}
\usepackage{myformat}

\newcommand{\Act}{\mathfrak A}
\newcommand{\Lu}{\mathscr L_u}
\newcommand{\Ru}{\mathcal R_u}
\newcommand{\mix}[3]{#1_{_{#2}}#3}
\newcommand{\mmix}[3]{#1_{_{\{#2\}}}#3}
\newcommand{\gv}{\gamma_{_V}}

\newcommand{\IV}{I_{_V}}
\newcommand{\Zu}{\mathcal Z_u}
\newcommand{\Ku}{\mathcal K_u}

\begin{document}

\mybic
\date \today 
 
\title
[Psychological Gambles]
{Subjective Expected Utility and Psychological Gambles}


\begin{abstract} 
We obtain an elementary characterization of expected utility 
based on a representation of choice in terms of psychological 
gambles, which requires no assumption other than coherence 
between ex-ante and ex-post preferences. Weaker version of 
coherence are associated with various attitudes towards 
complexity  and lead to a characterization of minimax or 
Choquet expected utility.

\end{abstract}

\maketitle

\section{Introduction}

An agent who is about to perform a given task, such as 
an exam, a competition or a choice, often speculates 
{\it ex-ante} on how the final result of his efforts will 
compare with his anticipations or with the outcome that 
he has decided to adopt as a reference. The mental 
process of entering into a tacit bet with oneself, to which 
we shall refer as a psychological gamble, often works as 
an attempt to mitigate the potential frustration that 
arises when falling short of a given standard. In the 
literature one finds illustrations of similar phenomena.
In an experimental study Sabater-Grande et al. 
\cite{sabater_et_al} document the importance of 
reference targets for university students facing exams. 
On the other hand, K\H{o}szegi and Rabin \cite{koszegi_rabin} 
develop a model of reference dependent choice in which 
the reference point reflects the agents expectations based
on their past experience in such standard choice situations 
as consumption choice or labour supply.

Although most examples of this phenomenon concerning
choice typically involve uncertainty on the consequences 
of our decisions, psychological gambles are actually involved 
in all choice problems. In fact even in such simple situations 
as planning a dinner out with friends, the exact details of
the context in which choice will take place, such as the set 
of alternatives actually available, their price, the interaction 
with other participants and so on, are not completely known 
in advance. In a recent work, Enke and Graeber 
\cite{enke_graeber} draw a sharp distinction between the
type of uncertainty that affects the consequences of choice
from what they refer to as {\it cognitive uncertainty} which
characterizes our ex-ante perception of choice and is often 
responsible of the noisiness of our decision process and of
deviations from the straightforward maximization of preferences.

Although this paper is not concerned with foundational 
issues, we take this short discussion as our motivation 
for reformulating choice problems in terms simple 
psychological gambles, which we introduce in Section 
\ref{sec gamble} as a subjectivist alternative to such 
objective randomizing mechanisms as horse lotteries of 
Anscombe and Aumann \cite{anscombe_aumann}.
Our starting point is that, notwithstanding the apparent 
conceptual differences between the behavioural process
of choice and the mind experiment of betting on our own 
anticipations, {\it any} choice problem may be equivalently 
written as a psychological gamble and assigning a utility 
value to some item $x$ is no different than assessing the 
probability that the anticipation of an outcome preferred 
to $x$ will be disappointed. The advantage of this reformulation
is that it embeds choice problems into the wider setting
of gambling strategies.

We apply psychological gambles to the problem of choice 
under uncertainty as originally described by Savage 
\cite{savage}, but taking the utility from acts and that 
from their consequences as two given functions. We 
provide a unified treatment of the most popular models 
in the theory of choice under uncertainty: subjective 
expected utility (SEU) of Savage \cite{savage}, multiple 
priors expected utility (MEU) of Gilboa and Schmeidler 
\cite{gilboa_schmeidler} and Choquet expected utility 
(CEU) of Schmeidler \cite{schmeidler}. These three 
models are characterized each one by a corresponding 
version of a single and simple property: coherence. This 
is an appropriate extension to gambles of the monotonicity 
axiom common to every model of decision theory under 
uncertainty and that, in plain terms suggests that gambles 
should be evaluated by the consequences they produce.

One merit of the setting presented hereafter is its relative
simplicity, although one may object that our approach is
more functional analytic than axiomatic. Another advantage
is the minimal assumptions we need concerning the structure 
of the set of acts. To our knowledge, this is the only paper 
not assuming that all simple acts are included in the choice 
set. Apart from mathematical generality, our approach
introduces new features such as preference bubbles,
see Theorem \ref{th GEU}, related to the possibility that 
some acts may induce consequences of unbounded utility.
Also, the classical strategies of proof, inspired either by
Savage \cite{savage} or by Anscombe and Aumann
\cite{anscombe_aumann}, are no longer possible under
our assumptions.

Many of our results exploit tools widely used in asset 
pricing theory. The connection with finance, first recognized 
by Gilboa and Samuelson \cite{gilboa_samuelson}, becomes 
fully clear in the context of gambling which provides a 
natural bridge between these apparently distinct areas 
of economic theory. One may conjecture that additional 
interesting implications may be drawn from the application
of these techniques. We prove e.g. in Corollary 
\ref{cor capacity} that every capacity is equivalent,
in an appropriate sense, to a subadditive one.

The paper is organized as follows. After describing the 
setting and our basic assumption \assref{structure}, in 
Section \ref{sec gamble} we introduce and discuss gambles.
In Section \ref{sec coherence} we introduce the crucial 
property of coherence and prove that this is necessary 
and sufficient for SEU. In Section \ref{sec arbitrage} we 
clarify the terms of the equivalence between coherence 
and arbitrage%
\footnote{
Arbitrage in turn is a concept strictly related to that of coherent 
betting scheme due to de Finetti \cite{de_finetti}, from which 
we borrow our terminology.
}
a central concept in asset pricing theory. In discussing the 
limits of coherence we stress that this property disregards 
the different level of complexity implicit in gambles and the 
role of individual aversion to complexity. In Section \ref{sec meu}, 
we thus discuss a weaker notion, $\Theta_0$-coherence,
and show that this is equivalent to MEU while in Section
\ref{sec ceu} we characterize CEU in terms of $\Theta_1$-%
coherence. Eventually, in Section \ref{sec capacity} we obtain 
some results on subjective capacities not involving coherence. 
We close in Section \ref{sec literature} with a short discussion 
of the existing literature.

\subsection{Notation}
The symbol  $\Fun{X,Y}$ indicates the family of all functions
$f:X\to Y$, $\Fun X$ is short for $\Fun{X,\R}$ while $\Fun[0]{X}$ 
designates the collection of reaal valued functions vanishing 
outside some finite set. On $\Fun[0]X$ we define the norm
$\norm f
	=
\sum_{x\in X}\abs{f(x)}$.
The symbol $ba(\A)$ (resp. $\Prob\A$, resp. $\Prob[0]\A$) 
denotes the family of bounded, finitely additive set functions 
(resp. probabilities, resp. probabilities with finite support) on 
some algebra of subsets of $\Omega$ and for each 
$\lambda\in ba(\A)$ we use the symbol $L^1(\lambda)$ 
for functions which are integrable with respect to $\lambda$. 
If $\mathcal H\subset\Fun\Omega$ we also write
\begin{equation}
\label{baF}
ba(\A,\mathcal H)
	=
\big\{\lambda\in ba(\A):
\mathcal H\subset L^1(\lambda)\big\}.
\end{equation}
Following de Finetti \cite{de_finetti}, a set is identified with its 
characteristic function, so that $A(x)$ is either $1$, if $x\in A$, 
or else $0$. If $A\subset X$ and $f,g\in\Fun X$ we write
$
\mix fAg
	=
f A+g A^c
$.

\section{The setting.}
\label{sec setting}

We shall use the following version of Savage setting.
$\Omega$ and $X$ are two completely arbitrary sets,
the former describing {\it states of nature}, the latter 
the space of {\it outcomes} (or prizes). The set $\Act$ 
of acts, as usual, consists of maps of $\Omega$ into 
$X$.  Following a well established tradition%
\footnote{
See e.g. Anscombe and Aumann \cite{anscombe_aumann}, 
Ghirardato et al. \cite{ghirardato_et_al} and, more recently, 
Gilboa and Samuelson \cite{gilboa_samuelson}.
},
we take two functions, $V\in\Fun\Act$ and $u\in\Fun X$, 
as given. The former describes preferences of the decision 
maker over acts, the second fixes the criterion to evaluate 
outcomes and is taken as a given constraint. In most
papers $u$ is just the restriction of $V$ to constant acts.

The functions $u$ and $V$ induce the order intervals
\begin{equation}
\label{I}
\IV(f)
	=
\{h\in\Act:V(h)\le V(f)\}
\qand
I_u(y)
	=
\{x\in X:u(x)\le u(y)\}.
\end{equation}
Associated with $u$ is also the linear space $\Lu$ spanned 
by the set $\{u(f):f\in\Act\}\subset\Fun\Omega$ as well as 
the collection $\Ru$ of all sets of the form $\{u(f)>t\}$ or 
their complements.

Concerning the structure of the set $\Act$ we assume the 
following:

\begin{Ass}
\label{structure}
The set $\Act$ contains two constant acts with values 
$x,y\in X$ such that%
\footnote{
The mixture operation on $\Fun{\Omega,X}$ is defined
by writing $\mix fAg=fA+gA^c$.
}:
(a) 
$V(y)>V(x)$,
(b)
$\mix yAf,\mix xAf\in\Act$ and 
(c)
$V(\mix yAf)\ge V(\mix xAf)$
for each $A\in\Ru$ and $f\in\Act$.
\end{Ass}

Notice that, by \assref{structure}, the collection 
\begin{equation}
\label{Au}
\A_u
	=
\{A\subset\Omega:\mix yAx\in\Act\}
\end{equation}
is an algebra containing $\Ru$.

The structure of the set $\Act$ as specified in \assref{structure}
is particularly poor and, borrowing terminology from Kopylov
\cite{kopylov_7}%
\footnote{
The meaning of small domain in \cite{kopylov_7} is, however,
quite different and refers to the restriction of subjective 
probability to a $\lambda$-system, rather than an algebra.
}, 
this qualifies $\Act$ as a ``small'' domain, in contrast with 
the full domain $\Fun{\Omega,X}$ adopted by Savage and 
by most of the following authors. Although the point may 
seem a minor one, it is not so, we believe, for several 
reasons leaving aside mathematical generality. First, if 
the state space $\Omega$ is infinite and $\Act$ coincides 
with the full domain $\Fun{\Omega,X}$, then the assumption 
that preferences are complete is extremely restrictive and 
hardly plausible from a behavioural point of view%
\footnote{
In his well known criticism, Aumann \cite{aumann} refers 
to completeness as ``perhaps the most questionable'' 
axiom.
}. 
Second, although property \tiref c above recalls Savage
sure thing principle, it is definitely a very weak version 
of it and most criticisms that apply to this classical
property, such as those arising from the Ellsberg 
paradox, do not apply here. Third, if $\Act$ does not 
contain all simple acts, the strategy of proof inaugurated
by Savage which requires to establish expected utility for 
all simple functions is simply not feasible%
\footnote{
Even in the subjective approach adopted by Ghirardato 
et al \cite{ghirardato_et_al} all simple acts must be
available.
}. 
Eventually, one cannot obtain $u$ as the restriction 
of $V$ to constant acts. In fact there is no self evident 
criterion to test whether $u$ and $V$ are compatible 
with one another and some explicit criteria will have
to be introduced. The lack of a direct link between the 
two utility functions may be appropriate to describe 
principal/agent situations in which, as suggested in 
\cite{gilboa_samuelson}, an agent chooses acts in 
accordance with the utility function $V$ while the 
principal evaluates the agent's choices by the 
consequences they produce.

Nevertheless, \assref{structure} permits to associate $V$ 
with the capacity
\begin{equation}
\label{capacity}
\gv(A)
=
\frac{V(\mix yAx)-V(x)}{V(y)-V(x)},
\qquad
A\in\A_u.
\end{equation}

\section{Psychological Gambles or 
Anscombe and Aumann need not 
play roulette.}
\label{sec gamble}

We introduced above the idea of a psychological bet on
$f$ in the following terms: taking $f$ as a reference, will 
the final outcome be dominated by $f$, i.e. will the event 
$\IV(f)$ occur? Viewed as a function,  $\IV(f)$ represents 
thus the payoff of the strategy $\delta_f$ which consists of 
betting one unit just on $f$. We generalize and formalize 
this intuitive description into the following:

\begin{definition}
A psychological gamble on $\Act$ (or simply a gamble) 
is an element of the set 
\begin{equation}
\label{Theta}
\Theta
	=
\Big\{\theta\in\Fun[0]{\Act,\R_+}:\sum_{f\in\Act}\theta(f)\le1\Big\}.
\end{equation}
\end{definition}

For given $\theta\in\Theta$, the quantity $\theta(f)$ describes 
the odds posted on the event $\IV(f)$ so that the gamble final 
payoff is%
\footnote{
For a comparison of \eqref{Pa} with the payoff of gambles 
in the approach of de Finetti see \cite{heath_sudderth_72}.
} 
\begin{equation}
\label{Pa}
\IV(\theta)
	=
\sum_{f\in\Act}\theta(f)\IV(f).
\end{equation}
The Dirac gambling strategy $\delta_f$ described above
induces the embedding of $\Act$ into $\Theta$ mentioned 
in the introduction. Abusing notation, we identify $f$ with 
$\delta_f$ and view acts as elementary gambles. Thus, 
the utility $V(f)$ of the act $f\in\Act$ is equivalently 
interpreted as the value $V(\delta_f)$ of the corresponding 
bet. This suggests a natural extension of the utility function 
$V$ to a linear functional on $\Theta$ defined as
\begin{equation}
\label{linear}
V(\theta)
	=
\sum_{f\in\Act}\theta(f)V(f),
\qquad
\theta\in\Theta
\end{equation}
which we interpret as the value of the gamble $\theta$. 
The next Lemma proves that this interpretation is consistent 
with the idea that the economic value of a gamble should 
reflect its perspective yields.

\begin{lemma}
\label{lemma value}
The value of a gamble is a linear function of its payoff.
\end{lemma}

\begin{proof}
It is easily seen that both maps, $\IV$ and $V$ defined in
\eqref{Pa} and \eqref{linear}, assign the same value to 
gambles $\theta$ and $\theta'$ which are equivalent 
according to the following criterion:
\begin{equation}
\label{modulo}
\sum_{f\sim g}\theta(f)
	=
\sum_{f\sim g}\theta'(f),
\qquad
g\in\Act.
\end{equation}
Thus, if $\IV(\theta)=\IV(\eta)$ we can assume, up to equivalence,
that the support of $\theta$ and $\eta$ are of the 
form $\{f_1,\ldots, f_n\}$ and $\{g_1,\ldots, g_k\}$ 
respectively with $V(f_i)<V(f_{i+1})$ and $V(g_j)<V(g_{j+1})$.
If $V(f_1)>V(g_1)$ then letting $h$ be the least valuable act 
between $g_2$ and $f_1$ we obtain
$\IV(\theta)(f_1)= \IV(\theta)(h)$ 
while 
$\IV(\eta)(f_1)\ne \IV(\eta)(h)$,
a contradiction. Thus necessarily $V(f_1)=V(g_1)$ and 
$\sum_{i\ge 1}\theta(f_i)
	=
\sum_{j\ge 1}\eta(g_j)$. 
Applying the same argument recursively we conclude that 
$n=k$ and that $V(f_p)=V(g_p)$ and 
$\sum_{i\ge p}\theta(f_i)
	=
\sum_{j\ge p}\eta(g_j)$
for each $p=1,\ldots,n$.
Thus, if $\IV(\theta)=\IV(\eta)$ then $\theta$ and $\eta$ must 
be equivalent in the sense of \eqref{modulo} so that
$V(\theta)
	=
V(\eta)$
which proves the claim.
\end{proof}

Of course, in purely mathematical terms, the set $\Theta$
of gambles is a mixture space. One may thus take it as the 
choice set, in place of $\Act$, and adopt the classical 
assumptions of Herstein and Milnor \cite{herstein_milnor}
to justify the linear formula \eqref{linear}. Our approach 
is rather to show that the expected utility representation 
is a property related to the {\it extension} of the original 
utility function from acts to gambles.

A gamble may be defined similarly on any partially 
ordered set. The special feature of $\Theta$ is that 
each gamble on $\Act$ may be considered from a 
{\it conditional} perspective, i.e. fixing the state 
$\omega$ and identifying an act $f$ with its 
consequence $f(\omega)$. The {\it conditional
value} of gamble $\theta$ is thus defined as
\begin{equation}
u(\restr{\theta}{\omega})
	=
\sum_{f\in\Act}\theta(f)u\big(f(\omega)\big),
\qquad
\theta\in\Theta
\end{equation}
and induces the partial order defined by writing
\begin{equation}
\label{ge_u}
\theta\ge_u\eta
\qiff
u(\restr\theta\omega)
	\ge
u(\restr\eta\omega)
\qtext{for all}
\omega\in\Omega.
\end{equation}

It may be instructive to compare conditional gambles with 
the set $\mathcal K$ of Anscombe and Aumann acts, i.e. the
set of {\it all} maps from $\Omega$ to $\Prob[0]X$. Indeed 
for fixed $\omega$ and with $\norm\theta=1$, a conditional 
gamble $\restr\theta\omega$ may be interpreted as a probability 
supported by the finite set 
$\{f(\omega):\theta(f)>0\}$%
\footnote{
The symbol $\mathcal K$ is taken from Fishburn 
\cite[p. 176]{fishburn_book}. Curiously, a probability distribution 
on $X$ with finite support -- i.e. an element of $\Prob[0]X$ -- is 
referred to as a gamble by Savage \cite[5.2]{savage}.
}. 
This would intrinsically imply an objectivist interpretation of
gambles which contrasts with our point of view. But even 
accepting this interpretation, it is important to remark that
conditional gambles form a strict subset of $\mathcal K$%
\footnote{
Using the axiom of choice, the set $\mathcal K$ has cardinality
$\abs X^{\abs\Omega}
\cdot
2^{\aleph_0\cdot\abs\Omega}$
while, even assuming $\Act=\Fun{\Omega,X}$, the cardinality
of $\Theta$ is just 
$\abs X^{\abs\Omega}
\cdot
2^{\aleph_0}$.
These two cardinals coincide if either $\Omega$ is a finite set 
or $\abs X\ge2^{\aleph_0}$.
}. 
In fact, upon revealing the true state $\omega$, nature only 
determines the {\it support} of the resulting distribution but 
not its size nor the probability weights. Given an arbitrary set
$P_1,\ldots,P_n\in\Prob[0]X$ and a partition $A_1,\ldots,A_n$
of $\Omega$, it is not possible, e.g., to construct an act 
$a\in\mathcal K$ such that $a=P_n$ on $A_n$, and it is enough to read 
the proof of the Anscombe and Aumann theorem with an 
infinite state space given by Fishburn 
\cite[Theorem 13.3]{fishburn_book}
to understand the importance for this theory of assuming such
a large domain as the choice set.

\section{Coherence and its limits.}
\label{sec coherence}

Given the lack of a direct relation between the two
functions $u$ and $V$, some general criteria to 
assess whether they are compatible may be 
formulated as follows:

\begin{definition}
\label{def coherence}
The utility functions $V$ and $u$ are said to be:
\begin{enumerate}[(a).]
\item
simply coherent with one another if  $f\ge_u g$ implies 
$V(f)\ge V(g)$, for every $f,g\in\Act$;
\item
coherent with one another if $\theta\ge_u \eta$ implies 
$V(\theta)\ge V(\eta)$, for every $\theta,\eta\in\Theta$%
\footnote{
In passing we remark that, restricting to gambles satisfying 
$\norm\theta=1$ would result in a definition of coherence 
too weak for our purposes.
}.
\end{enumerate}
\end{definition}

To appreciate the difference between these two conditions
notice that, given \assref{structure}, simple coherence implies
$u(y)>u(x)$ while, assuming coherence, we obtain $u(x)=tV(x)$ 
and $u(y)=tV(y)$ for some $t>0$. This in turn implies that 
coherence is preserved if we replace $u$ with the transform 
$A+Bu$ and $V$ with $A+tBV$, provided $B>0$. In other 
words, with no loss of generality we shall henceforth assume
\begin{equation}
\label{normalize}
V(y)=u(y)=1
\qand
V(x)=u(x)=0.
\end{equation}

Simple coherence is a basic condition, ubiquitous in this
literature and roughly corresponds to ({\bf P7}) of Savage, 
whereas coherence is its natural extension from acts to 
gambles. While the latter condition is equivalent to SEU, 
as we will show in Theorem \ref{th GEU}, the former is 
way to poor to deliver a significant representation of utility. 
Nevertheless, under additional special assumptions the 
two conditions coincide. Ghirardato et al \cite{ghirardato_et_al}, 
e.g., assume that $\Act$ contains all constant acts and 
that $V$ (and therefore $u$) has convex range. Thus, for 
each $\theta\in\Theta$ there exists $f_\theta\in\Act$ 
such that $V(\theta)=V(f_\theta)$ which makes the two 
preceding conditions identical.

If $f$ is any act, we shall often use Assumption 
\assref{structure} to obtain the truncations:
\begin{equation}
\label{truncate}
f^k=\mix f{\{u(f)\le k\}}x,\quad
f_k=\mix f{\{u(f)> -k\}}x
\qand
f(k)=\mix f{\{-k<u(f)\le k\}}x,
\qquad
k\ge0.
\end{equation}

\begin{theorem}
\label{th GEU}
Assume \assref{structure}. The utility functions $V$ and 
$u$ are coherent if and only if they satisfy
\begin{equation}
\label{GEU}
V(f)
	=
\Phi\big(u(f)\big)+\int_\Omega u(f)dm,
\qquad
f\in\Act
\end{equation}
in which 
(a)
$\Phi$ is a positive linear functional on $\Lu$  
vanishing on bounded functions and 
(b)
$m\in\Prob{\A_u,\Lu}$.
The representation \eqref{GEU} is unique.
\end{theorem}

\begin{proof}
If $V$ and $u$ satisfy \eqref{GEU} they are clearly 
coherent. Conversely, assume that $u$ and $V$ are
coherent. By linearity the extension of $V$ and $u$ 
to gambles may be stretched still from $\Theta$ to 
$\Fun[0]\Act$ by letting
\begin{equation}
V(\zeta)
	=
\sum_{f\in\Act}\zeta(f)V(f)
\qand
u(\restr\zeta\omega)
	=
\sum_{f\in\Act}\zeta(f)u(f(\omega))
\qquad
\zeta\in\Fun[0]\Act.
\end{equation}
For fixed $f\in\Act$, $f_1=\mmix f{u(f)>u(x)}x$ and
$f_2=\mmix f{u(f)\le u(x)}x$ are elements of $\Act$;
moreover, $u(f_1)=u( f)^+$ and $u(f_2)=- u( f)^-$.
Therefore, letting $\beta_f\in\Fun[0]\Act$ be such 
that $\beta_f(f_1)=1$, $\beta_f(f_2)=-1$ and $\beta_f=0$
elsewhere, we obtain that
\begin{equation*}
u(\beta_f)
=
\babs{u( f)}.
\end{equation*}
Thus, for each $\zeta\in\Fun[0]\Act$ there exists 
$\zeta^*\in\Fun[0]\Act$, defined by letting 
$\zeta^*(h)
=
\sum_{f\in\Act}\abs{\zeta(f)}\beta_f(h)$,
such that
\begin{equation}
\label{directed}
\babs{u(\restr\zeta\omega)}
	\le
\sum_{f\in\Act}\babs{\zeta(f)u(f(\omega))}
	\le
\sum_{f\in\Act}\babs{\zeta(f)}\sum_h\beta_f(h)u(h(\omega))
	=
u(\restr{\zeta^*}\omega).
\end{equation}

Moreover, if $V(\zeta)<0$ then, letting 
$\bar\zeta
	=
\zeta/(\norm\zeta\vee1)$,
$\bar\zeta^+,\bar\zeta^-\in\Theta$ and 
$V(\bar\zeta^+)<V(\bar\zeta^-)$, by coherence, there exists 
$\omega\in\Omega$ such that 
$u(\restr{\bar\zeta^+}\omega)
	<
 u(\restr{\bar\zeta^-}\omega)$
i.e.
\begin{equation}
\label{conglo}
\inf_{\omega\in\Omega}u(\restr\zeta\omega)
	<
0.
\end{equation}

Given \eqref{directed} and \eqref{conglo}, it follows from 
\cite[Theorem 3.3]{JCA_2018} the existence of a positive 
linear functional $\Phi$ on $\Lu$ which vanishes on bounded 
functions and of $\bar m\in ba(\Omega)_+$ such that
\begin{align}
\label{Th}
u(\theta)\in L(\bar m)
\qand
V(\theta)
	=
\Phi\big(u(\theta)\big)
+
\int_\Omega u(\theta)d\bar m,
\qquad
\theta\in\Theta.
\end{align}
In view of \eqref{normalize}, $\bar m$ is a probability. Moreover,
the restriction $m$ of $\bar m$ to $\A_u$ coincides with$\gv$
and is thus unique. For fixed $f\in\Act$, it follows from the 
inclusion $u(f)\in L^1(m)$ that
\begin{equation}
\label{m}
\int u(f)dm
	=
\lim_k\int_{\{\abs{u(f)}\le k\}}u(f)dm
	=
\lim_k\int u\big( f(k)\big)dm
	=
\lim_kV(f(k))
\end{equation}
and therefore
\begin{equation}
\label{Phi}
\Phi(u(f))
	=
\lim_kV(\mix f{\{\abs{u(f)}>k\}}x).
\end{equation}
This proves uniqueness.
\end{proof}

If the representation \eqref{GEU} is perhaps not a surprising 
result {\it per se}, some of its features deserve some comments. 
First we highlight that the crucial integrability condition 
$\Lu\subset L^1(m)$ is endogenous and not an assumption.
This is remarkable since, under \assref{structure} the random 
quantity $u(f)$ need not be bounded. In Savage setting, the 
function $u$ has to be bounded on $X$ as a consequence of 
the existence of a countable partition of $\Omega$ into non 
null sets and of the full domain assumption
$\Act=\Fun{\Omega,X}$ (see \cite[Theorem 14.5]{fishburn_book}). 
Likewise, in the extension of Anscombe and Aumann model to 
an infinite state space \cite[Theorem 13.3]{fishburn_book}, 
boundedness of the function $u(f)$ for all acts $f\in\Act$ 
is a direct consequence of the assumption that $\mathcal K$ 
contains {\it all} possible horse/roulette lotteries%
\footnote{
See the proof of properties $S3$ and $S5$ \cite[p. 183-185]
{fishburn_book}. 
}.
Similar conclusions apply to much of this literature. The 
unbounded part of $u(f)$ is explicitly evaluated by the 
functional $\Phi$ which may thus rightfully be interpreted 
as a {\it bubble in preferences}. This is a novel feature in 
the theory of decision making under uncertainty. The 
economic role of bubbles emerges in those situations 
in which $f\ge_ug$ but $f(k)\not\ge_ug(k)$.

Eventually, we observe that coherence is necessary and 
sufficient for \eqref{GEU}, i.e. that indeed SEU is a property 
involving gambles, although this is not always apparent in 
most approaches. The set-up of Anscombe and Aumann is, 
in contrast, not at all necessary.

We can dispose of preference bubbles with the aid
of the following continuity axiom:

\begin{Ass}
\label{A_cts}
For every $f,g,h\in\Act$ with $V(g)>V(f)>V(h)$ there exists 
$h\in\Act$ with $V(h)>V(x)$ such that $V(\mix yAx)<V(h)$ 
implies $V(h)<V(\mix f{A^c}x)<V(g)$ for each $A\in\A_u$.
\end{Ass}

This property, although not necessary for our results, makes 
the comparison with the literature easier. It is easily seen 
that, given \assref{structure}, assumption \assref{A_cts}  
implies
\begin{equation}
\lim_{\gv(A)\to0}V(\mix f{A^c}x)
	=
V(f),
\qquad
f\in\Act.
\end{equation}
In \eqref{GEU} the inclusion $\Lu\subset L^1(m)$ implies 
$\gv(\abs{u(f)}>k)=m(\abs{u(f)}>k)\to0$. Thus, under
\assref{A_cts}
\begin{align*}
V(f)
	=
\lim_kV( f(k))
	=
\lim_k\int_{\{\abs{u(f)}\le k\}}u(f)dm
	=
\int u(f)dm,
\qquad
f\in\Act
\end{align*}
i.e. $\Phi=0$. Conversely, if \eqref{GEU} holds with $\Phi=0$
and if $m(A_n)\to0$ then
\begin{align*}
V(f)
	=
\int u(f)dm
	=
\lim_n\int_{A_n^c}u(f)dm
	=
\lim_n\int u(\mix f{A_n^c}x)dm
	=
\lim_nV(\mix f{A_n^c}x).
\end{align*}

\begin{corollary}
\label{cor GEU}
Assume \assref{structure}. The utility functions $V$ and $u$ 
are coherent and \assref{A_cts} is satisfied if and only  if the 
representation \eqref{GEU} holds with $\Phi=0$.
\end{corollary}

\section{Coherence and arbitrage.}
\label{sec arbitrage}

In Theorem \ref{th GEU} the utility function $V$ was taken
as given although coherence need not be preserved when
considering a different representation f the same preference 
system. We want now to consider the more general problem 
of whether, among the many equivalent representations of 
the same preference system, there exists a utility function 
coherent with $u$. This problem has a surprisingly natural 
translation into the language of asset pricing and, in particular, 
into the notion of an arbitrage opportunity (or lack of). 

Let us start defining the set
\begin{equation}
\label{Z}
\Zu
	=
\big\{\zeta\in\Fun[0]\Act: 0\ge_u\zeta\big\},
\end{equation}	
defining $\mathcal P_+$ (resp. $\mathcal P_{++}$) as the 
convex cone spanned by the functions of the form
$\delta_b-\delta_a$ 
where $a,b\in\Act$ and $V(b)\ge V(a)$ (resp. $V(b)>V(a)$)
and writing 
$\Ku	= \Zu-\mathcal P_+$
\footnote{
Notice that the definition of $\Ku$ is independent of the
utility function $V$.
}.

It is easily seen that $V$ and $u$ are simply coherent 
if and only if%
\footnote{
Condition \eqref{NA} is formally identical to the 
mathematical definition of absence of arbitrage 
opportunities in mathematical finance. Compare
with \cite{kreps}.
}
\begin{equation}
\label{NA}
\Ku\cap\mathcal P_{++}
	=
\emp.
\end{equation}
Coherence requires a more stringent condition than
\eqref{NA}, involving the topology on $\Fun[0]\Act$
generated by sets of the form 
\begin{equation}
\label{basis}
U_{H,a}
	=
\Big\{\eta\in\Fun[0]\Act:\sum_{f\in\Act}\eta(f)H(f)<a\Big\},
\qquad
a\in\R,\ 
H\in\Fun\Act.
\end{equation}

\begin{theorem}
\label{th NFLVR}
The utility function $V$ admits an equivalent representation
coherent with $u$ if and only if 
\begin{equation}
\label{NFLVR}
\cl{\mathcal K}_u\cap\mathcal P_{++}
	=
\emp.
\end{equation}
\end{theorem}

\begin{proof}
If the utility function $V$ on $\Act$ is coherent with $u$, 
then $V(\kappa)\le0$ for every $\kappa\in\mathcal K_u$ and, 
by \eqref{basis}, the same inequality extends to the closure 
$\cl{\mathcal K}_u$ so that \eqref{NFLVR} holds. Conversely, 
fix $f_0,g_0\in\Act$ such that $V(f_0)>V(g_0)$, let $\phi_0$ 
be a continuous linear functional on $\Theta$ such that
\begin{equation}
\phi_0(\delta_{f_0}-\delta_{g_0})
	>
0
	\ge
\sup_{\zeta\in\cl{\mathcal K}_u}\phi_0(\zeta)
\end{equation}
and define
\begin{equation}
\label{Vphi}
V_0(h)
	=
\phi_0(h),
\qquad
h\in\Act.
\end{equation}
Then, $V_0(f_0)>V_0(g_0)$ while, if $V(b)\ge V(a)$, 
the inclusion $\delta_a-\delta_b\in\mathcal K_u$ 
implies
$\phi_0(\delta_a-\delta_b)
	\le
0$ 
and so $V_0(a)\le V_0(b)$. We can then normalize 
$\phi_0$ so that $\abs{V_0}\le1$ on some arbitrary, 
non empty order interval $\IV(h_0)$. Form a sequence
$\seqn h$ in $\Act$ such that $V(h_n)$ increases to
$\sup_{f\in\Act}V(f)$.

Consider the family 
$\mathscr J
=
\{\IV(f)\setminus \IV(g):g,f\in\Act,\ V(f)>V(g)\}$. It admits
a countable sub collection
$\mathscr J_0
	=
\{\IV(f_n)\setminus \IV(g_n):n\in\N\}$ 
with the property that every $J\in\mathscr J$ admits some
$J_n\in\mathscr J_0$ such that $J_n\subset J$%
\footnote{
This is an easy consequence of the fact that $\mathscr J$
contains at most countably many singletons and that the
set $\{V(f):f\in\Act\}$ is separable.
}.

We can repeat the same step for each $n\in\N$, 
with $h_n$ in place of $h_0$ and replacing $g_0,f_0$
with $f_n,g_n$ such that 
$\IV(f_n)\setminus \IV(g_n)
	\in
\mathscr J_0$.
For each iteration we obtain a linear functional $\phi_n$ 
separating  $\delta_{f_n}-\delta_{g_n}$  from 
$\cl{\mathcal K}_u$ and such that the function 
$V_n$ defined as in \eqref{Vphi} satisfies $V_n(a)\le V_n(b)$ 
whenever $V(a)\le V(b)$ and $\abs{V_n}\le1$ on $\IV(h_n)$. 
Let 
$\phi
	=
\sum_n2^{-n}\phi_n$ 
and again define $V_\phi$ as in \eqref{Vphi}. For each 
$f\in\Act$ there exists $N$ such that 
$f\in \IV(h_N)$. Thus,
\begin{align*}
\abs{V_\phi(f)}
	\le
\sum_n2^{-n}\abs{V_n(f)}
	\le
1+\sum_{n\le N}2^{-n}\abs{V_n(f)}
	<
+\infty.
\end{align*}
Moreover, if $V(f)>V(g)$ and if 
$\IV(f_n)\setminus\IV(g)
	\subset
\IV(f)\setminus\IV(g)$ then
\begin{align*}
V_\phi(f)-V_\phi(g)
	=
\phi(\delta_f-\delta_g)
	\ge
2^{-n}\phi_n(\delta_{f_n}-\delta_{g_n})
	>
0.
\end{align*}
In addition, if $V(b)\ge V(a)$ then, as seen above,
$V_n(b)\ge V_n(a)$ for all $n\in\N$. Then, $V_\phi$ 
is a utility function. Moreover if 
$\theta,\eta\in\Theta$ and $\theta\ge_u\eta$ then 
$\eta-\theta\in\mathcal K_u$
so that $V_\phi(\theta)\ge V_\phi(\eta)$ which proves 
the claim.
\end{proof}

The technique used in Lemma \ref{th NFLVR} is known 
in mathematical finance as Yan Theorem (see \cite{yan} 
or \cite{AMAS_2007}) and it will be used again in the 
next sections. Notice that $V$ plays virtually no role in
the preceding result (save in proving the existence of the 
countable subset $\mathscr J_0$ mentioned in the proof)
and it may thus be restated exclusively in terms of preferences.

\section{Multiple priors and maxmin preferences.}
\label{sec meu}

We shall now start discussing some popular decision 
models in terms of deviations from coherence. In so 
doing it is useful to formulate some minimal 
restrictions in order to rule out pathological situations.  

\begin{definition}
\label{def minimal}
The utility functions $V$ and $u$ are compatible if
there exists a function $c\in\Fun{\Theta,\R_+}$ such that
\begin{equation}
\label{minimal}
c(t\theta)\le tc(\theta)
\quad
0\le t\le1
\qand
\theta\ge_u\eta
\qtext{implies}
V(\theta)+c(\theta)\ge V(\eta),
\qquad
\eta\in\Theta.
\end{equation}
\end{definition}

The quantity $c(\theta)$ describes the cost of complexity 
inherent in the gamble $\theta$. Definition \ref{def minimal}
suggests that deviations from coherence are due to
the need to adjust the value of a gamble by its cost,
which is assumed to be finite. 

\begin{definition}
\label{def relative coherence}
Let  $\hat\Theta\subset\Theta$. The utility functions 
$u$ and $V$ are $\hat\Theta$-coherent if they are 
compatible and if the intervening cost function 
$c$ vanishes on $\hat\Theta$.
\end{definition}

Intuitively, the gambles included in $\hat\Theta$ are
considered by the decision maker as particularly simple
so as to bear no cost of complexity. The most natural
candidate are the gambles involving, further to $y$,
only one non constant act. Their set is denoted by 
$\Theta_0$.

\begin{theorem}
\label{th MEU}
Assume \assref{structure} and \assref{A_cts}. The utility 
functions $V$ and $u$ are $\Theta_0$-coherent if and 
only if there exists
$\Lambda\subset\Prob{\A_u}$ such that
\begin{equation}
\label{MEU}
\Lu\subset\bigcup_{\lambda\in\Lambda}L^1(\lambda)
\qand
V(f)
	=
\inf_{\lambda\in\Lambda}\int_\Omega u(f)d\lambda
\qquad
f\in\Act.
\end{equation}
\end{theorem}

In \eqref{MEU} whenever $u(f)\notin L^1(\lambda)$ we
define conventionally $\int u(f)d\lambda=+\infty$, as
customary.

\begin{proof}
Recalling \eqref{Z}, define the convex cone 
$\Xi
	=
\bigcup_{\lambda>0}\lambda\Theta-\Zu$
and the extended real-valued functional
\begin{equation}
\kappa(\xi)
	=
\sup\big\{\lambda V(\eta):
\lambda>0,\ \eta\in\Theta,\ \xi\ge_u\lambda\eta\big\},
\qquad
\xi\in\Xi,
\end{equation}
which is clearly positively homogeneous. If $\xi\ge_u\lambda\eta$ 
and $\bar\xi=\xi/(\norm\xi\vee1)$, then
\begin{equation*}
\frac{\bar\xi^+}{1+\lambda}
\ge_u
\frac{\lambda}{(1+\lambda)(\norm\xi\vee1)}\eta
+
\frac{1}{1+\lambda}\bar\xi^-
\equiv
\eta_1
\end{equation*}
and $\eta_1,\bar\xi^+\in\Theta$. If $u$ and $V$ are 
compatible, there exists a cost function $c$
such that
\begin{align*}
\frac{\lambda}{(1+\lambda)(\norm\xi\vee1)}V(\eta)
+
\frac{1}{1+\lambda}V(\bar\xi^-)
=
V(\eta_1)
\le
\frac{V(\bar\xi^+)+c(\bar\xi^+)}{1+\lambda}.
\end{align*}
This in turn implies
\begin{align}
\kappa(\xi)
	\le
V(\xi)+c(\bar\xi^+)(1\vee\norm\xi)
	<
+\infty.
\end{align}
Therefore $\kappa$ is real valued, superadditive and 
monotonic with respect to $\ge_u$; moreover $\kappa\ge V$ 
on $\Theta$, while $\kappa=V$ on $\Theta_0$.

For fixed $f\in\Act$, define the $f$ translate $\Xi_f= \Xi-\{f\}$
and the functional
\begin{equation}
\label{xi}
\kappa_f(\beta)
	=
\kappa(\beta+f)-\kappa(f)
\qquad
\beta\in\Xi_f.
\end{equation}
$\Xi_f$ is a convex set containing $\Xi$ while $\kappa_f$ is 
\tiref i
concave (and thus such that $\kappa_f\ge\kappa$ on $\Xi$), 
\tiref{ii}
$\ge_u$ monotonic, 
\tiref{iii} 
such that $-\kappa_f(-f)
	=
V(f)
	=
\kappa_f(f)$
and
\tiref{iv}
$\kappa_f(t\delta_y)=t$
for any $t\ge0$.

Implicit in these properties is the fact that if $\beta\in\Xi_f$ 
is such that 
$s_u(\beta)
	\equiv
\sup_\omega u(\restr\beta\omega)
	<
+\infty$,
then
$\kappa_f(\beta)\le s_u(\beta)$. This conclusion follows from
\tiref{ii} and either \tiref{iv}, if $s_u(\beta)\ge0$, or else \tiref i,
via the inequalities
\begin{align*}
0
\ge
\kappa_f\Big(\frac{\beta-s_u(\beta)\delta_y}{1-s_u(\beta)}\Big)
\ge
\frac{\kappa_f(\beta)}{1-s_u(\beta)}
+
\frac{-s_u(\beta)}{1-s_u(\beta)}.
\end{align*}

Then, if $\beta=\sum_{i=1}^nw_i\beta_i$ is 
a convex combination of elements from 
$\Xi_f^*=\{\beta\in\Xi_f:s_u(\beta)<+\infty\}$
we deduce
\begin{align*}
\sum_{i=1}^nw_i\kappa_f(\beta_i)
	\le
\kappa_f(\beta)
	\le
s_u(\beta).
\end{align*}
As a consequence of \cite[Theorem 4.1]{JCA_2023} there
exists $\lambda_f\in\Prob{\A_u}$ such that 
\begin{equation}
u(\beta)\in L^1(\lambda_f)
\qand
\kappa_f(\beta)
	\le
\int u(\beta)d\lambda_f,
\qquad
\beta\in\Xi^*_f.
\end{equation}
Let $\Lambda=\{\lambda_f:f\in\Act\}$.

The inclusions $f^k\in\bigcap_{h\in\Act}\Xi^*_h$ and
$-f_k\in\bigcap_{n\le k}\Xi^*_{f_n}$ lead to the
conclusion 
\begin{equation}
u(f^k)
\in
\bigcap_{\lambda\in\Lambda}L^1(\lambda)
\qand
-u(f_k)
\in
\bigcap_{n\le k}L^1(\lambda_{f_n})
\end{equation}
so that $u(f)\in\bigcap_n L^1(\lambda_{f_n})$ and thus
$\Lu
	\subset
\bigcup_{\lambda\in\Lambda}L^1(\lambda)$.
Moreover,
\begin{align}
\label{V(f^k)}
V(f^k)
	=
\inf_{h\in\Act}\kappa_h(f^k)
	\le
\inf_{\lambda\in\Lambda}\int u(f^k)d\lambda
	\le
\inf_{\lambda\in\Lambda}\int u(f)d\lambda
	<
+\infty.
\end{align}
On the other hand,
\begin{align}
\label{V(f_k)}
V(f_k)
	=
-\kappa_{f_k}(-f_k)
	\ge
\int u(f_k)d\lambda_{f_k}
	\ge
\inf_{\lambda\in\Lambda}\int u(f)d\lambda.
\end{align}
The representation then follows from \assref{A_cts} through
the implication $V(f)=\lim_kV(f^k)=\lim_kV(f_k)$.


Conversely, let $\Lambda$ satisfy \eqref{MEU} and 
choose $\theta,\eta\in\Theta$ such that $\theta\ge_u\eta$. 
Then,
\begin{align*}
V(\eta)
	=
\sum_{h\in\Act}\eta(h)\inf_{\lambda\in\Lambda}
\int u(h)d\lambda
	\le
\inf_{\lambda\in\Lambda}\int u(\eta)d\lambda
	\le
\inf_{\lambda\in\Lambda}\int u(\theta)d\lambda
	=
V(\theta) +c(\theta)
\end{align*}
where we have implicitly set
$c(\theta)
	=
\inf_{\lambda\in\Lambda}\int u(\theta)d\lambda
-
V(\theta)
	\ge
0$. If, in addition, $\theta\in\Theta_0$ so that
$\theta=a\delta_f+b\delta_y$ for some $a,b\in\R_+$,
then $V(\theta)
=
\inf_{\lambda\in\Lambda}\int u(\theta)d\lambda$
so that $c(\theta)=0$.
\end{proof}

A version of the preceding result may be established without
invoking \assref{A_cts}.

\begin{corollary}
\label{cor MEU}
Assume \assref{structure}. If the utility functions $V$ and $u$ 
are $\Theta_0$-coherent then there exist 
(a)
a concave functional $\Phi\in\Fun{\Lu}$ which 
vanishes on bounded functions and 
(b)
a convex set $\Lambda\subset\Prob{\A_u}$ such that 
\begin{equation}
\label{MMEU}
\Lu\subset\bigcup_{\lambda\in\Lambda}L^1(\lambda)
\qand
V(f)
	=
\Phi(u(f))
+
\inf_{\lambda\in\Lambda}\int_\Omega u(f)d\lambda
\qquad
f\in\Act.
\end{equation}
\end{corollary}

\begin{proof}
We start noting that, by coherence, \eqref{MMEU} may be
taken as a definition of $\Phi$. Retaining from the proof
of Theorem \ref{th MEU} the inequalities \eqref{V(f^k)}
and \eqref{V(f_k)}, we conclude that
\begin{align}
V(f)-V(f^k)
	\ge
\Phi(u(f))
	\ge 
V(f)-V(f_k)
\end{align}
so that indeed $\Phi$ vanishes on bounded elements of $\Lu$
\end{proof}

As we drop \assref{A_cts} the representation \eqref{MMEU} 
ceases to be sufficient for $\Theta_0$-coherence as it was 
not possible to prove that $\Phi$ is concave. This confirms the 
importance even in this new setting of the assumptions inducing
boundedness.

\section{Choquet expected utility.}
\label{sec ceu}

In a gamble
$\theta\in\Theta_0$ there is essentially only one source
of randomness. One may extend this principle and consider
positions in several different acts which are ``similar'' to
one another in that the utility of the consequences they 
produce are affected by randomness in quite the same way. 
A popular criterion is comonotonicity and it is obvious that 
all gambles in $\Theta_0$ are formed by comonotonic acts. 
By $\Theta_1$ we thus designate the set of gambles $\theta$ 
such that any two acts, $f,g$, in the support of $\theta$ are 
comonotonic in the sense that
\begin{equation}
\label{comonotonic}
\big[u(f(\omega))-u(f(\omega'))\big]
\big[u(g(\omega))-u(g(\omega'))\big]
	\ge
0,
\qquad
\omega,\omega'\in\Omega.
\end{equation}

\begin{theorem}
\label{th CEU}
Assume \assref{structure} and 
\assref{A_cts}. Then $V$ and $u$ are $\Theta_1$-coherent 
if and only if $\gv$ is a convex capacity satisfying the 
Choquet expected utility representation
\begin{equation}
\label{CEU}
V(f)
	=
\int u(f)d\gv,
\qquad
f\in\Act.
\end{equation}
\end{theorem}

\begin{proof}
The proof is elementary and relies on a stepwise
approximation of $u(f)$. For each $k,n\in\N$ define
\begin{equation}
h_i
	=
\mix y{\{u(f)\ge i2^{-n}-k\}}x,
\qquad
i=1,\ldots,k2^{n+1}-1
\end{equation} 
and let $\eta_n^k(h_i)=2^{-n}$ or else $\eta_n^k(g)=0$.
Then $\norm{\eta_n^k}=2k-2^{-n}$ and
\begin{align*}
2^{-n}u(y)+\sum_{g\in\Act}\eta_n^k(g)u(g)
	\ge
u(f(k))+ku(y)
	\ge
\sum_{g\in\Act}\eta_n^k(g)u(g).
\end{align*}
Then,
\begin{equation}
\label{sandwich}
\frac{2^{-n}\delta_y+\eta^k_n}{2k}
	\ge_u
\frac{f(k)+k\delta_y}{2k}
	\ge_u
\frac{\eta^k_n}{2k}.
\end{equation}
Given that all three terms in \eqref{sandwich} are elements 
of $\Theta_1$ we conclude
$V(f(k))
	\ge 
V(\eta_k^n)-k
	\ge 
V(f(k))-2^{-n}$.
In addition,
\begin{align*}
\lim_nV(\eta_n^k)
	&=
\lim_n\sum_{i=1}^{k2^{n+1}-1}2^{-n}\gv\big(u(f)\ge i2^{-n}-k\big)
	\\&=
\int_0^k\gv(u(f)\ge t-k)dt
	\\&=
\int_{-k}^k\gv(u(f)\ge t)dt
	\\&=
\int_{-k}^k u(f)d\gv+k.
\end{align*}
From this we deduce \eqref{CEU} by letting $k\to+\infty$ 
and exploiting \assref{A_cts}. 

Let $A,B\in\A_u$. Since $\mix y{\cap B}x$ and $\mix y{A\cup B}x$ 
are comonotonic, the gamble $\theta$ which bets $1/2$ on each 
of these two acts is in $\Theta_1$. Let $\eta$ be the gamble which 
bets $1/2$ on $\mix yAx$ and $\mix yBx$. Given $u(\theta)=u(\eta)$ 
we conclude by $\Theta_1$-coherence
\begin{align*}
\gv(A\cup B)+\gv(A\cap B)
	=
2V(\theta)
	\ge 
2V(\eta)
	=
\gv(A)+\gv(B)
\end{align*}
proving convexity

If, conversely, the Choquet representation \eqref{CEU} 
holds, $\theta\ge_u\eta$ and $\theta\in\Theta_1$ then
by ordinary rules of the Choquet integral,
$V(\theta)
	=
\int u(\theta)d\gv
	\ge
\int u(\eta)d\gv
	\ge
V(\eta)$.
\end{proof}

The importance of Choquet expected utility formula
\eqref{CEU}, introduced in the theory of decisions by 
Schmeidler \cite{schmeidler_89}, justifies some historical 
comments. In his seminal contribution Schmeidler 
\cite{schmeidler} rightly attributes a special importance 
to convex capacities in view of two facts:
\tiref{i}
these are the infimum of the elements of their core and 
\tiref{ii}	
the corresponding Choquet integral is superadditive and 
additive over comonotonic functions. Following Schmeidler, 
the credit for discovering these two important properties 
has thenceforth unanimously been ascribed to Shapley 
\cite{shapley_71} (the paper was originally published in 
1965) and to Dellacherie \cite{dellacherie}, respectively. 
As a matter of fact (and quite curiously) both results were 
already well known  at the time the above references were 
published since they had been proved, and in more general 
terms, by Eisenstadt and Lorentz \cite{eisenstadt_lorentz} 
in 1959 (see Theorem 2 in that paper which fully anticipates
the main result of \cite{schmeidler}). Even the term 
{\it convex} (or rather concave), attributed to Shapley,
had already been introduced by Eisenstadt and Lorentz.

\section{The structure of subjective capacities.}
\label{sec capacity}

In the preceding sections we have repeatedly used 
the capacity $\gv$ defined in \eqref{capacity}, which, 
depending on the degree of coherence assumed, may 
or not be additive or convex. A different utility 
representation of the same preference system will in 
general induce a capacity with different properties.
The following result provides another illustration of
the connection between decision theory and finance.
Notice that a weaker property than \assref{structure}
is actually needed as we only use mixing of $x$ and $y$.

\begin{lemma}
\label{lemma submeasure}
Under \assref{structure} there
exists a utility function $V_*$ on $\Act$ which is equivalent 
to $V$ and is associated with a subadditive capacity.
\end{lemma}

\begin{proof}
For simplicity we adopt the notation 
$\IV(g,f)
=
\IV(f)\setminus \IV(g)$.
Consider a function $H\in\Fun\Act$ of the form
\begin{equation}
\label{H}
H
	=
\sum_{n=1}^Nt_n\big\{\IV(\mix y{B_n}x,\mix y{A_n\cup B_n}x)
-
\IV(x,\mix y{A_n}x)\big\}
\end{equation}
where $t_n\ge 0$, $\gv(B_n)\ge\gv(A_n)$ and, upon 
rearranging terms, $\gv(A_N)\ge\ldots\ge\gv(A_1)>0$. 
Then,
\begin{align*}
\min H+\max H
	\le
H(\mix y{A_1}x)+\max H
	=
-\sum_{n=1}^Nt_n+\max H
	\le
0.
\end{align*}
This inequality applies {\it a fortiori} to all elements of the
convex cone of functions on $\Act$ dominated by some $H$ as in
\eqref{H} as well as to the closure of such cone in the
topology of uniform distance. As in the proof of Lemma
\ref{th NFLVR} we can apply Yan Theorem and obtain
the existence of some $\mu\in ba(\Act)_+$ such that
\begin{equation}
\sup_{\{A,B\in\A_u:A\subset B\}}
\mu\big(\IV(\mix yBx,\mix y{A\cup B}x)-\IV(x,\mix yAx)\big)
	\le
0
	<
\mu(\IV(g,f)),
\qquad
f,g\in\Act,\ 
V(f)>V(g).
\end{equation}
Since $\mu\ne0$ we can normalize $\mu$ so that
$\mu\big(\IV(x,y)\big)=1$. The set function 
$\gamma_*(A)
=
\mu(\IV(x,\mix yAx))$ 
possesses then the following properties, valid for all 
$A,B\subset\Omega$:
\begin{equation}
\label{submeasure}
(a)\ 
\gamma_*(\emp)=0,\quad 
(b)\ 
\gamma_*(A)\le\gamma_*(B) \text{ when }A\subset B
\qand 
(c)\ 
\gamma_*(A\cup B)
	\le
\gamma_*(A)+\gamma_*(B).
\end{equation}
Define
\begin{equation}
V_*(f)
	=
\mu\big(\IV(x,f)\big)-\mu\big(\IV(f,x)\big),
\quad
f\in\Act,
\end{equation} 
so that $V_*(\mix yAx)=\gamma_*(A)$ and
\begin{align*}
V^*(f)-V^*(g)
	=
\mu\big(\IV(g,f)\big)-\mu\big(\IV(f,g)\big),
\qquad
g,f\in\Act.
\end{align*}
Thus $V^*$ is a utility function equivalent to $V$
and such that $V_*(x)=0$ and $V_*(y)=1$.
\end{proof}

A set function satisfying \eqref{submeasure} is often called
a {\it submeasure}. It is possible to reformulate the preceding 
result in a purely measure theoretic language.

\begin{corollary}
\label{cor capacity}
Every capacity $\gamma$ on an algebra $\A$ of subsets 
of $\Omega$ admits a submeasure $\gamma_*$ on $\A$ 
such that
\begin{equation}
\gamma(A)>\gamma(B)
\qiff
\gamma_*(A)>\gamma_*(B),
\qquad
A,B\in\A.
\end{equation}
\end{corollary}

\begin{proof}
Take the space of bounded, real valued, $\A$ measurable 
functions as the collection of acts and for each such function
$b$ let $V(b)=\int bd\gamma$ and apply Lemma 
\ref{lemma submeasure}.
\end{proof}

 One
interesting property of $\gv$ is connected with the following 
axiom.

\begin{Ass}
\label{A_partition}
For every set $B\in\A_u$ such that $\gv(B)>0$ there exists 
$N\in\N$ with the property that any $\A_u$ measurable 
partition $\pi$ of $\Omega$ of size $\ge N$ satisfies
\begin{equation}
\gv(B\setminus A)
	>
\gv(A)
\qqtext{for some}
A\in\pi.
\end{equation}
\end{Ass}

\begin{lemma}
\label{lemma KR}
Assume \assref{structure}. Then
\assref{A_partition} is equivalent to the existence of
$P\in\Prob{\A_u}$ such that 
\begin{equation}
\label{KR}
\lim_nP(A_n)=0
\qiff
\lim_n\gv(A_n)=0.
\end{equation} 
\end{lemma}

\begin{proof}
It is obvious that $\gv$ satisfies either \assref{A_partition} or 
\eqref{KR} if and only if so does the submeasure $\gamma_*$ 
defined in Lemma \ref{lemma submeasure}. Fix then $B$ such 
that $\gamma_*(B)>0$ and assume \assref{A_partition}. 
Then, there exists $N\in\N$ such that for every disjoint 
collection $A_1,\ldots,A_N$ of subsets of $\Omega$ 
we can find an index $1\le i\le N$ such that 
$\gamma_*(B\setminus A_i)
	>
\gamma_*(A_i)$
and thus 
$\gamma_*(A_i)
	<
\gamma_*(B\setminus A_i)
	\le
\gamma_*(B)$.
This inequality, being true for all $B$ with $\gamma_*(B)>0$,
implies that for each $\varepsilon>0$ any partition
of sufficiently large size contains an element $A$ 
such that $\gamma_*(A)<\varepsilon$. In other words, 
$\gamma_*$ is uniformly exhaustive and it thus satisfies 
\eqref{KR} for some $P\in\Prob\Omega$, by 
\cite[Theorem 3.4]{kalton_roberts}.

Conversely, assume \eqref{KR} and fix $B$ such that
$\gv(B)>0$. Choose 
$0
	<
3\varepsilon
	<
\gv(B)$ 
and let $P(A)$ is sufficiently small so that
$\gv(A)
	\le
\varepsilon$
and
$\gv(B)
	\le 
\gv(B\setminus A)+\varepsilon$, 
which we may do by choosing appropriately $A$ from any
partition of sufficiently large size. Then,
$
\gv(A)
	<
\gv(A)+\varepsilon
	<
\gv(A)+\gv(B)-2\varepsilon
	\le
\gv(B)-\varepsilon
	\le
\gv(B\setminus A)$.
\end{proof}

It is natural to compare \assref{A_partition} with the 
apparently similar axiom ({\bf P6}) of Savage. In plain 
terms, \assref{A_partition} asserts that any large partition 
of $\Omega$ contains {\it at least one} ``small'' (possibly 
null) set. Savage ({\bf P6}) requires instead the existence 
of a partition (necessarily of large size) {\it all} of whose 
elements are equally ``small'' (and thus with positive mass). 
Despite its relative weakness, assumption \assref{A_partition} 
is sufficient to establish a relation between utility from 
acts with probability. Although such probability may not 
be unique it need not be atomless either, thus avoiding 
an artificial implication of Savage construction.

\section{Relation with the literature}
\label{sec literature}

Introducing a convex structure on the set of acts, e.g.  
via the embedding of acts into gambles, is, of course, 
not a new idea. Most of the papers which adopt it, 
however, end up with justifications which are mainly 
of an objectivist nature, such as the roulette lotteries
of Anscombe and Aumann \cite{anscombe_aumann}. 
This raises the criticism that introducing objective elements into a purely subjective 
framework is somewhat contradictory. Scott \cite{scott} 
presents a totally abstract (but extremely appealing) 
framework to choice based on formal sums which can 
be applied to several problems involving the numerical 
representation of binary relations over a finite set, such 
as non transitive preferences. His paper is perhaps the 
first attempt to obtain a choice theoretic model from 
simple separating arguments and was later followed 
by Fishburn \cite{fishburn_75} who considered the 
problem of choice among probability distributions. The 
best attempt to obtain a linear structure with purely 
behavioural assumptions is the paper by Ghirardato 
et al. \cite{ghirardato_et_al} in which the concept of 
{\it utility mixtures} is introduced based on the assumption 
that the image of $X$ under $u$ is convex. But even in
this setting one needs to assume that acts include all 
simple functions.

In relatively recent times a number of papers have
looked at decision problems from a financial point
of view and are thus indirectly related to our
interpretation of utility as the value of a gamble. 
Echenique et al. \cite{echenique_et_al} propose
to test revealed preference theory by measuring the 
profits that could be made in case of violations of 
its axioms. Echenique and Saito \cite{echenique_saito} 
implement a similar approach for expected utility
under risk aversion. The paper that comes closer to 
the present one is, however, the recent work of Gilboa 
and Samuelson \cite{gilboa_samuelson} in which,
working in a finite setting, it is investigated whether 
the utility functions $u$ and $V$ satisfy the integral 
representation (to which they refer to as coherence) 
or the maxmin representation. They interpret this as 
a problem of an agent trying to justify his choices, based 
on $V$, to some principal who evaluates the results 
through $u$. They assume that acts include all functions
$\Omega\to X$ and do not provide a behavioural interpretation
of the results obtained which are mainly geometric in
nature. However, it is the first paper that recognizes 
explicitly the connection with finance.

\BIB{abbrv}

\end{document}